\documentclass{llncs}
\usepackage{tikz}
\usetikzlibrary{arrows,shapes}
\usepackage{mathpartir}
\usepackage{bigpage}
\usepackage{bcprules}
\usepackage{mathtools}
\usepackage{listings}
\usepackage{amsmath}
\usepackage{amssymb}
\usepackage{comment}
\usepackage{hyperref}
\usepackage{longtable}
\usepackage{stmaryrd}

\newcommand{\new}{\mathsf{new}}

\newcommand{\pic}{$\pi$-calculus}

\newcommand{\ldb}{[\![}
\newcommand{\rdb}{]\!]}

\newcommand{\wbbisim}{\stackrel{\centerdot}{\approx}} 

\newcommand{\id}[1]{\texttt{#1}}

\newcommand{\pzero}{\mathbin{0}}

\newcommand{\juxtap}{\mathbin{\id{|}}}

\newcommand{\scong}{\mathbin{\equiv}}
\newcommand{\nameeq}{\mathbin{\equiv_N}}

\newcommand{\names}[1]{\mathbin{\mathcal{N}(#1)}}
\newcommand{\freenames}[1]{\mathbin{\mathsf{FN}(#1)}}
\newcommand{\boundnames}[1]{\mathbin{\mathsf{BN}(#1)}}
\newcommand{\binpar}[2]{#1 | #2}
\newcommand{\outputp}[2]{#1!(#2)}
\newcommand{\prefix}[3]{\mathsf{for}(#2 \leftarrow #1?) #3}

\newcommand{\quotep}[1]{\mathsf{@}#1}
\newcommand{\dropn}[1]{\mathsf{*}#1}
\newcommand{\procn}[1]{\stackrel{\vee}{x}}

\newcommand{\substn}[2]{\{ #1 / #2 \}}

\newcommand{\meaningof}[1]{\ldb #1 \rdb}

\newcommand{\Proc}{\mathsf{Proc}}
\newcommand{\QProc}{\quotep{\mathsf{Proc}}}

\newcommand{\bc}{\mathbin{\mathbf{::=}}}
\newcommand{\bm}{\mathbin{\mathbf\mid}}

\newcommand{\red}{\rightarrow}
\newcommand{\wred}{\Rightarrow}

\newcommand{\vect}[1]{\stackrel{\rightharpoonup}{#1}}

\newcommand{\rhoc}{$\rho$-calculus}

\makeatletter
\gdef\tshortstack{\@ifnextchar[\@tshortstack{\@tshortstack[c]}}
\gdef\@tshortstack[#1]{%
  \leavevmode
  \vtop\bgroup
    \baselineskip-\p@\lineskip 3\p@
    \let\mb@l\hss\let\mb@r\hss
    \expandafter\let\csname mb@#1\endcsname\relax
    \let\\\@stackcr
    \@ishortstack}
\makeatother

\title{Name-free combinators for concurrency}
\author{
L.G. Meredith\inst{1}\\
\and
Michael Stay\inst{2}\\
}
\institute{
  {RChain Cooperative}\\
  \email{\fontsize{8}{8}\selectfont greg@rchain.coop}
  \and
  {Pyrofex Corp.}\\
  \email{\fontsize{8}{8}\selectfont stay@pyrofex.net}\\
}
\begin{document}
\maketitle
\begin{abstract}
\noindent
  Yoshida demonstrated how to eliminate the bound names coming from
  the input prefix in the asynchronous {\pic}, but her combinators
  still depend on the $\new$ operator to bind names.
  We modify Yoshida's combinators by replacing $\new$ and replication
  with reflective operators to provide the first combinator calculus
  with no bound names into which the asynchronous {\pic} has a faithful
  embedding.  
\end{abstract}

\section{Introduction}

Many term calculi, like $\lambda$-calculus or {\pic}, involve binders
for names, and the mathematics of bound variable names is
subtle. Sch\"onfinkel introduced the SKI combinator calculus in 1924
to clarify the role of quantified variables in intuitionistic logic by
eliminating them \cite{finkel}; Curry developed Sch\"onfinkel's ideas
much further. The difficulties are not merely theoretical, but
represent a real practical challenge in the design of programming
languages. In fact binding is one of the key features of
the PoPLMark Challenge \cite{PoPLMark}. Certainly, the recent work by Jamie
Gabbay and Andrew Pitts \cite{DBLP:journals/fac/GabbayP02} and others
\cite{DBLP:journals/jcss/Clouston14} on nominal set theory has put the
study of bound names and substitution on a much nicer foundation that
can be shown to extend to practical implementations. However, it
introduces an intriguing conundrum.

Specifically, the work of Gabbay and Pitts relies on a version of set
theory (Fraenkel-Mostowski set theory, which we abbreivate to FM set
theory) that admits an infinite supply of ``atoms''. This raises the
question of where these atoms come from, which has both theoretical
and practical implications. On the practical side, infinite sets of
``atomic'' entities, i.e. entities with no internal structure, are not
realizable on modern computers. Modern computers fundamentally rely on
sets of elements with effective internal structure to provide the kind
of compression necessary to produce or compute with infinite sets. The
natural numbers is a prime example. Because of their very regular
internal structure the entire set can be represented by a single
recursive equation. Potentially then, the natural numbers or some
other effectively representable set could provide the source of atoms
used in an FM-set theoretic account of binding in practical
implementations. However, this raises a new question.

In order to represent and compute these effectively representable sets
a notion of computation must already be in place. If that notion of
computation relies on a notion of binding, then not a lot progress has
been made! As in \cite{DBLP:journals/entcs/MeredithR05} we argue that
this circularity, instead of being an obstacle to overcome, might be a
clue to an alternative approach to binding phenomena; and one that
ties together two important computational phenomena that have not
normally been considered as related. More explicitly, we extend the
argument made by Meredith and Radestock that reflection suffices to
provide the ``atoms'' used in the {\pic} as channels to produce the
first name-free set of combinators that enjoys a full and faithful
interpretation of the calculus.

While the focus of the paper is largely on the technical results it is
useful to consider the larger context motivating them. Reflection and
meta-programming features more generally are part of a growing number
of mainstream languages. Java, C\#, even the Haskell and OCaml
communities have seen growing interest in these features with efforts
like template Haskell and MetaOCaml, respectively. In large measure
this has to do the fact that programming at industrial scale requires
the leverage of computer programs to write computer programs. Thus,
meta-programming features are simply a practical necessity. On the
other hand, reflection and meta-programming have not received a
theoretical account that fits well with strong typing. This is one of
the reasons why language designs based on typed $\lambda$-calculi have
been so slow to adopt reflection as a feature by comparison to other
language designs.

In this setting, the idea that a single feature already enjoying
widespread adoption could account for such a subtle phenomenon as the
kind of binding found in the {\pic} is both intriguing and worth
exploring, even if FM set theory provides a satisfying account of
computation with nominal phenomena in other respects. However, it is
precisely the foundational theoretical questions raised by the FM set
theoretic approach that motivates the investigation in the first place:
what better place to look for the source of ``atoms'' than in the
reification of theory of computation requiring them? 

To be clear, we are not focusing on ordinary abstraction, as that is
well solved by abstraction elimination from the $\lambda$-calculus to
$\mathcal{SKI}$. Instead we are focused on binders for fresh
names. For example, the {\pic} (\cite{milner91polyadicpi}) has two
binders for names: the $\new$ operator, which introduces a new name
into scope, and the input prefix, which introduces a name for labeling
locations for substitution.  Yoshida
\cite{DBLP:journals/tcs/Yoshida02} describes an elimination algorithm
that gets rid of input prefixes which corresponds in many respects to
the elimination of lambda abstraction; but her combinators still
fundamentally depend on the $\new$ operator.  In complementary work,
Meredith and Radestock \cite{DBLP:journals/entcs/MeredithR05}
introduce reflective operators into a higher-order {\pic} and implment
$\new$ and replication in terms of reflection.  Here, we present a
fusion of those ideas: a name-free concurrent combinator calculus into
which Yoshida's combinators have a faithful embedding.

\subsubsection{Role of compositionality in theories of concurrency} 
It is also interesting how compositionality plays out in this
setting. As we emphasize throughout the paper, the {\pic\;} is not a
closed theory, but one depending on a theory of names. In other words,
a fully specified theory of processes involves a composition of
theories, namely \footnote{pun gratefully accepted} an application of
the {\pic\;} with some theory of names. In this sense we derive the {\rhoc}
from the fixed point of applying the theory of processes to itself
(regarded as a theory of names). Symbolically,

\begin{equation*}
  \mathcal{R} = \Pi[\quotep{\mathcal{R}}]
\end{equation*}

where $\Pi$ produces a theory of processes from whatever it is given
as a theory of names and $\quotep{\mathcal{R}}$ is the quoted forms of
said processes, regarding them as names.

This is exactly the kind of design level thinking that
compositionality should promote, and it is actually surprising that
this particular solution for a name-free version of a concurrent
calculus wasn't found sooner. In point of fact, the initial
implementations of {\rhoc} in \textsf{OCaml}, \textsf{Haskell}, and
\textsf{Scala} all used exactly this fixed point formulation to define the
syntax for the {\rhoc}. The fact that the code type-checks
and that the types are inhabited provides some assurance of the
soundness of the construction.

The combinator versions take compositionality a step further by
removing even more of the syntax, effectively arriving at a variant of
an applicative algebra over a small handfull of combinators to account
for all binding and mobility phenomena. Further still, the soundness
of the semantics makes essential use of compositionality because we
are effectively composing Yoshida's original semantics with the
reflective account of name construction and deconstruction; thus
illustrating that compositionality is not just for design-level
thinking, but provides powerful compression of proofs.

\subsubsection{Outline of the paper}
To be fully self-contained this paper would need to present four
different calculi and two different encodings: the original {\pic},
Yoshida's combinator calculus, and the encoding from the term calculus
to the combinator; the {\rhoc}, and the encoding from the {\pic\;} to
the {\rhoc} and the new reflective combinator calculus, and the
encoding from the term calculus to the reflective combinator
calculus. Such a manifest provides all the technical inventory to
illustrate how the two encoding techniques, prefix elimination and
$\new$ elimination, combine and how the encoding from the {\pic\;} can
be constructed by composing the encoding into the {\rhoc} with the
encoding into the reflective combinator calculus. We have provided the
complete manifest, but have pushed the presentation of the {\rhoc} and
the {\pic\;} and the encoding of the latter into the former in an
appendix at the end as the {\pic\;} is quite well known at this point,
and the {\rhoc} has been part of the concurrency literature for over a
decade. This organization allows us to focus on the newer results of a
name-free combinator calculus for mobile concurrency.

\subsubsection{Related work}

As mentioned previously, there is a long history of interest in
combinatorial presentations of computational models, with Yoshida's
work representing a seminal development for concurrent computation. We
have also recently become aware of
\cite{DBLP:journals/toplas/RajaS97}. They achieve a similar goal, but
without reflection, using a technique pioneered by Quine, who was also
a pioneer of reflective techniques. \cite{Quine59} \cite{Quine60}

\section{A reflective higher-order concurrent combinator calculus}

\subsection{Yoshida's original combinator calculus}

This is the briefest account of Yoshida's original calculus while her
paper, which is nearly 20 years old, is still an outstanding piece of
research and well worth the effort to read in conjunction with these results.

\begin{mathpar}
  \inferrule* [lab=atom] {} { P \bc 0 \;|\; \mathsf{m}(a,b) \;|\; \mathsf{d}(a,b,c) \;|\; \mathsf{k}(a) \;|\; \mathsf{fw}(a,b) \;|\; \mathsf{b}_{\mathsf{r}}(a,b) \;|\; \mathsf{b}_{\mathsf{l}}(a,b) \;|\; \mathsf{s}(a,b,c) }
  \and
  \inferrule* [lab=process] {} {\bm \; (\new\; a)P \;|\; P|P \;|\; \mathsf{*}P}
\end{mathpar}

We write $\mathsf{c};\mathsf{c'};\ldots;$ to denote these
agents. $\mathsf{m}(a,b)$ (message) carries $a$ name $b$ to name $a$,
$\mathsf{d}$ (duplicator) distributes a message to two locations,
$\mathsf{fw}$ (forwarder) forwards a message (thus linking two
locations), $\mathsf{k}$ (killer) kills a message, while
$\mathsf{b}_{\mathsf{r}}$ (right binder), $\mathsf{b}_{\mathsf{l}}$
(left binder) and $\mathsf{s}$ (synchroniser) generate new links. In
particular $\mathsf{b}_{\mathsf{r}}$ and $\mathsf{b}_{\mathsf{l}}$
represent two different ways of binding names – in
$\mathsf{b}_{\mathsf{r}}$ one uses the received name for output, while
in $\mathsf{b}_{\mathsf{l}}$ one uses it for input. In contrast,
$\mathsf{s}$ is used for pure synchronisation without value passing,
which is indeed necessary in interaction scenarios.

As in the {\pic}, the $\new$ operator is a binding operator
for names, so Yoshida's calculus also has a notion of free and bound names.

\begin{mathpar}
  \freenames{\pzero} := \emptyset
  \and
  \freenames{\mathsf{k}(a)} := \{ a \}
  \and
  \freenames{\mathsf{m}(a,b)} = \freenames{f(a,b)} = \freenames{\mathsf{b}_{\mathsf{r}}(a,b)} = \freenames{\mathsf{b}_{\mathsf{l}}(a,b)} := \{ a, b \}
  \and
  \freenames{\mathsf{d}(a,b,c)} = \freenames{\mathsf{s}(a,b,c)} := \{ a, b, c \}
  \and
  \freenames{(\new\;a)P} := \freenames{P} \setminus \{ a \}
  \and
  \freenames{P|Q} := \freenames{P} \cup \freenames{Q}
  \and
  \freenames{!P} := \freenames{P}
\end{mathpar}

The bound names of a process, $\boundnames{P}$, are those names occurring in $P$
that are not free. For example, in $(\new\; b)\mathsf{m}(a,b)$, the name $a$ is free, while $b$ is bound.

In the following definition, $\vect{x}$ indicates a list of names,
$u:\vect{x}$ indicates the concatenation of $u$ onto the vector, and
abuse set notation $u \in \vect{x}$ to assert or require that $u$
occurs in $\vect{x}$.

\begin{definition}
Two processes, $P,Q$, are alpha-equivalent if $P = Q\{\vect{y}/\vect{x}\}$ for
some $\vect{x} \in \boundnames{Q},\vect{y} \in \boundnames{P}$, where $Q\{\vect{y}/\vect{x}\}$
denotes the capture-avoiding substitution of $\vect{y}$ for $\vect{x}$ in $Q$.
\end{definition}

\begin{definition}
  The {\em structural congruence} $\equiv$
  between processes \cite{SangiorgiWalker} is the least congruence containing
  alpha-equivalence and satisfying the commutative monoid laws
  (associativity, commutativity and $\pzero$ as identity) for parallel
  composition $|$.
\end{definition}

Rewrite rules
\[\begin{array}{rl}
  \mathsf{d}(a,b,c) | \mathsf{m}(a,x) & \red \mathsf{m}(b,x) | \mathsf{m}(c,x) \\
  \mathsf{k}(a) | \mathsf{m}(a,x) & \red 0 \\
  \mathsf{fw}(a,b) | \mathsf{m}(a,x) & \red \mathsf{m}(b,x) \\
\end{array} \quad \quad
\begin{array}{rl}
  \mathsf{b}_{\mathsf{r}}(a,b) | \mathsf{m}(a,x) & \red \mathsf{fw}(b,x) \\
  \mathsf{b}_{\mathsf{l}}(a,b) | \mathsf{m}(a,x) & \red \mathsf{fw}(x,b) \\
  \mathsf{s}(a,b,c) | \mathsf{m}(a,x) & \red \mathsf{fw}(b,c)
\end{array}\]
\[\begin{array}{rl}
  \mathsf{*}P & \red P|\mathsf{*}P \\
\end{array}\]
\begin{mathpar}
  \inferrule* {{P} \red {P}'} {{{P} | {Q}} \red {{P}' | {Q}}}
  \and
  \inferrule* {{{P} \scong {P}'} \andalso {{P}' \red {Q}'} \andalso {{Q}' \scong {Q}}}{{P} \red {Q}}
\end{mathpar}

\subsubsection{Translating the {\pic} into Yoshida's combinators}
We assume the following annotations ($+$ stands for the output and $-$
stands for the input), which denote how each name is used in the rules
of interaction:

\[\mathsf{m}(a^{+},v^{\pm});\mathsf{d}(a^{-},b^{+},c^{+});\mathsf{k}(a^{-});\mathsf{fw}(a^{-},b^{+});\mathsf{b}_{\mathsf{l}}(a^{-}b^{+});\mathsf{b}_{\mathsf{r}}(a^{-}b^{-});\mathsf{s}(a^{-}b^{-}c^{+})\]

Note the annotated polarities are preserved by reduction, e.g.
\[\binpar{\mathsf{d}(a^{-},b^{+},c^{+})}{\mathsf{m}(a^{+},v)} \to \binpar{\mathsf{m}(b^{+},v)}{\mathsf{m}(b^{+},v)}\]

It is worth pointing out that we use a slightly different syntax for input-guarded processes. Where most readers familiar with $x?(y)P$ we write $\mathsf{for}(y \leftarrow x)P$, not only to be more instep with modern programming languages, but also because it generalizes more naturally to join constructions, such as $\mathsf{for}(y_1 \leftarrow x_1; \ldots; y_n \leftarrow x_n)P$.

\begin{mathpar}
  \meaningof{a!(b)} = \mathsf{m}(a,b)
  \and
  \meaningof{\mathsf{for}(x \leftarrow a)P} = \mathsf{for}^*(x \leftarrow a)\meaningof{P}
  \and
  \meaningof{P|Q} = \meaningof{P}|\meaningof{Q}
  \and
  \meaningof{(\mathsf{new}\;x)P} = (\mathsf{new}\;x)\meaningof{P}
  \and
  \meaningof{\mathsf{*}P} = \mathsf{*}\meaningof{P}
  \and
  \meaningof{0} = 0
\end{mathpar}

where

\[\begin{array}{llrl}
(I)&\mathsf{for}^*(x \leftarrow a)(P|Q) &=& (\mathsf{new}\; c_1c_2)(\mathsf{d}(a,c_1,c_2) | \mathsf{for}^*(x \leftarrow c_1)P | \mathsf{for}^*(x \leftarrow c_2)Q)\\
(II)&\mathsf{for}^*(x \leftarrow a)(\mathsf{new}\; c')P &=& (\mathsf{new}\; c)\mathsf{for}^*(x \leftarrow a)P\substn{c}{c'} \\
(III)& \mathsf{for}^*(x \leftarrow a)\pzero &=& \mathsf{k}(a) \\
(IV)&\mathsf{for}^*(x \leftarrow a)\mathsf{*}P &=& (\mathsf{new}\;c)(\mathsf{fw}(a,c) | \mathsf{*}\mathsf{for}^*(x \leftarrow c)(P | \mathsf{m}(c,x)))\\
(V)&\mathsf{for}^*(x \leftarrow a)c(v^+,\vec{w}) &=& (\mathsf{new}\;c)(\mathsf{s}(a,c,v) | c(c^+,\vec{w})) \mbox{\quad $x \notin \{v\vec{w}\}$}\\
(VI)&\mathsf{for}^*(x \leftarrow a)c(v^-,\vec{w})&=& (\mathsf{new}\;c)(\mathsf{s}(a,v,c) | c(c^-,\vec{w})) \mbox{\quad $x \notin \{v\vec{w}\}$}\\
(VII)&\mathsf{for}^*(x \leftarrow a)\mathsf{m}(v,x)&=& \mathsf{fw}(a,v) \mbox{\quad $x \neq v$}\\
(VIII)&\mathsf{for}^*(x \leftarrow a)\mathsf{fw}(x,v)&=& \mathsf{b}_{\mathsf{l}}(a,v) \mbox{\quad $x \neq v$}\\
(IX)&\mathsf{for}^*(x \leftarrow a)\mathsf{fw}(v,x)&=& \mathsf{b}_{\mathsf{r}}(a,v) \mbox{\quad $x \neq v$}\\
(X)&\mathsf{for}^*(x \leftarrow a)c(\vec{v}_1x^+\vec{v}_2)&=&(\mathsf{new}\; c)\mathsf{for}^*(x \leftarrow a)(\mathsf{fw}(c,x) | c(\vec{v}_1c\vec{v}_2)) \mbox{\quad $x \notin \vec{v}_1$}\\
(XI)&\mathsf{for}^*(x \leftarrow a)c(x^-\vec{v})&=& (\mathsf{new}\; c)\mathsf{for}^*(x \leftarrow a)(\mathsf{fw}(x,c) | c(c\vec{v}))\\
(XII)&\mathsf{for}^*(x \leftarrow a)\mathsf{b}_{\mathsf{r}}(v,x^-)&=& (\mathsf{new}\; c_1c_2c_3)\mathsf{for}^*(x \leftarrow a)(\mathsf{d}(v,c_1,c_2)|\mathsf{s}(c_1,x,c_3)|\mathsf{b}_{\mathsf{r}}(c_2,c_3)) \mbox{\quad $x \notin \{v \}$}\\
(XIII)&\mathsf{for}^*(x \leftarrow a)\mathsf{s}(v,x^-,w) &=& (\mathsf{new}\; c_1c_2)\mathsf{for}^*(x \leftarrow a)(\mathsf{s}(v,c_1,c_2) | \mathsf{m}(c_1,x) | \mathsf{b}_{\mathsf{l}}(c_2,w)) \mbox{\quad $x \notin \{v \}$}
\end{array}\]

\subsection{Reflective higher-order (RHO) combinator calculus}
The {\pic} is not a closed theory, but rather a theory dependent upon
some theory of names. Taking an operational view, one may think of the
{\pic} as a procedure that when handed a theory of names provides a
theory of processes that communicate over those names. This openness
of the theory has been exploited in {\pic} implementations like the
execution engine in Microsoft's Biztalk \cite{biztalk}, where an
ancillary binding language provides a means of specifying a `theory'
of names: {\em e.g.}, names may be TCP/IP ports, or URLs, or object
references, {\em etc.}  But foundationally, one might ask if there is
a closed theory of processes, {\em i.e.} one in which the theory of
names arises from and is wholly determined by the theory of
processes. Meredith and Radestock have shown that this is not only
possible, but results in a calculus that enjoys both the features of
concurrency and meta-programming
\cite{DBLP:journals/entcs/MeredithR05}. The key idea is to provide the
ability to quote processes, effectively reifying them as names, and to
unquote them, effectively reflecting names back as processes.

The same technique can be applied to Yoshida's combinators. We remove
new names and replication, introduce quoting/unquoting
operators. Notice that this effectively allows processes in the second
argument of a send, because the name in the second argument is
`merely' a quoted process. This affords an opportunity to introduce an
extra rewrite governing the interaction between sending and unquoting.

\begin{mathpar}
  \inferrule* [lab=atom] {} { P \bc 0 \;|\; \mathsf{m}(a,\quotep{P}) \;|\; \mathsf{d}(a,b,c) \;|\; \mathsf{k}(a) \;|\; \mathsf{fw}(a,b) \;|\; \mathsf{b}_{\mathsf{r}}(a,b) \;|\; \mathsf{b}_{\mathsf{l}}(a,b) \;|\; \mathsf{s}(a,b,c) }
  \and
  \inferrule* [lab=process] {} {\bm \; *a \;|\; P|P}
  \and
  \inferrule* [lab=nominal] {} {a \bc \quotep{P}}
\end{mathpar}

Rewrite rules
\[\begin{array}{rl}
\mathsf{d}(a,b,c) | \mathsf{m}(a,\quotep{P}) & \red \mathsf{m}(b,\quotep{P}) | \mathsf{m}(c,\quotep{P}) \\
\mathsf{k}(a) | \mathsf{m}(a,\quotep{P}) & \red 0 \\
\mathsf{fw}(a,b) | \mathsf{m}(a,\quotep{P}) & \red \mathsf{m}(b,\quotep{P}) \\
\mathsf{b}_{\mathsf{r}}(a,b) | \mathsf{m}(a,\quotep{P}) & \red \mathsf{fw}(b,\quotep{P}) \\  
\end{array} \quad \quad
\begin{array}{rl}
  \mathsf{b}_{\mathsf{l}}(a,b) | \mathsf{m}(a,\quotep{P}) & \red \mathsf{fw}(\quotep{P},b) \\
  \mathsf{s}(a,b,c) | \mathsf{m}(a,\quotep{P}) & \red \mathsf{fw}(b,c) \\
  *(a) | \mathsf{m}(a,\quotep{P}) & \red P
\end{array}\]
\begin{mathpar}
  \inferrule* {{P} \red {P}'} {{{P} | {Q}} \red {{P}' | {Q}}}
  \and
  \inferrule* {{{P} \scong {P}'} \andalso {{P}' \red {Q}'} \andalso {{Q}' \scong {Q}}}{{P} \red {Q}}
\end{mathpar}

\begin{definition}
  The {\em structural congruence} $\equiv$
  between processes \cite{SangiorgiWalker} is the least congruence
  satisfying the commutative monoid laws
  (associativity, commutativity and $\pzero$ as identity) for parallel
  composition $|$ and $*(@(P)) \equiv P$.
\end{definition}

Note that alpha equivalence is no longer part of structural
congruence.  While there is a faithful embedding of Yoshida's
combinators into RHO combinators (see below), RHO combinators can see
the internal structure of names and distinguish them.

\subsubsection{Implementing replication with reflection}

%
%
%
%
%
As mentioned before, it is known that replication (and hence
recursion) can be implemented in a higher-order process algebra
\cite{SangiorgiWalker}. As our first example of calculation with the
machinery thus far presented we give the construction explicitly in
the RHO combinator calculus.

\begin{definition}[Replication]
  \label{replication}
  $D(x,v,w) := (\mathsf{d}(x,v,w) | \mathsf{fw}(v,x) | {*}(w))$
\end{definition}
\[\begin{array}{rl}
  \mathsf{*}_{(x,v,w)} P &= \mathsf{m}(x,\quotep{(D(x,v,w) \; |\; P)}) \; |\; D(x,v,w) \\
        &= \mathsf{d}(x,v,w) \; |\; \mathsf{fw}(v,x) \; |\; {*}(w) \; |\; \mathsf{m}(x,\quotep{(D(x,v,w) \; |\; P)}) \\
        &\red \mathsf{m}(v,\quotep{(D(x,v,w) \; |\; P)}) \; |\; \mathsf{m}(w,\quotep{(D(x,v,w) \; |\; P)}) \; |\; \mathsf{fw}(v,x) \; |\; {*}(w) \; |\; \mathsf{m}(x,\quotep{(D(x,v,w) \; |\; P)}) \\
        &\red \mathsf{m}(x,\quotep{(D(x,v,w) \; |\; P)}) \; |\; \mathsf{m}(w,\quotep{(D(x,v,w) \; |\; P)}) \; |\; {*}(w) \\
        &\red \mathsf{m}(x,\quotep{(D(x,v,w) \; |\; P)}) \; |\; D(x,v,w) \; |\; P \\
        & = \; \mathsf{*}_{(x,v,w)} P \; |\; P
\end{array}\]

Of course, this encoding, as an implementation, runs away, unfolding
$\mathsf{*}P$ eagerly. It is possible to obtain a lazier
replication operator restricted to the embedding of
input-guarded $\pi$ processes. The reader familiar with the
$\lambda$-calculus will have noticed the similarity between $D$ and
the ``paradoxical'' or ``fixed point'' combinator $Y$.

\subsubsection{Implementing new names with reflection}

Here we provide an encoding of Yoshida's combinator calculus into the
RHO combinator calculus. Since all names are global in the RHO
combinator calculus, we encounter a small complication in the
treatment of free names at the outset. There are several ways to
handle this.  One is to insist that the translation be handed a closed
program, one in which all names are bound either by input or by
restriction, but this feels inelegant. Another is to provide
an environment, $r : \mathcal{N}_{\mbox{\tiny Yoshida}} \rightarrow \QProc$, for
mapping the free names in a Yoshida process into names in the RHO
combinator calculus. Maintaining the updates to the environment,
however, obscures the simplicity of the translation. We adopt a third
alternative.

To hammer home the point that Yoshida's combinator calculus is
parameterized in a theory of names, we instantiate her calculus with
the names of the RHO combinator calculus. This is no different than
instantiating her calculus using the natural numbers, or the set of URLs as the
set of names. Just as there is no connection between the structure of
these kinds of names and the structure of processes in the {\pic},
there is no connection between the processes quoted in the names used
by the theory and the processes generated by the theory, and we
exploit this fact.

Let $\QProc$ be set of names in the RHO combinator calculus, $\Proc$
be the set of terms of the RHO combinator calculus, and
$\Proc_{\mbox{\tiny Yoshida}}$ be the set of terms of her combinator
calculus built using $\QProc$ \emph{as the names}. The translation will be
given in terms of a function \[\meaningof{-}_2( -, - ) : 
    \Proc_{\mbox{\tiny Yoshida}} \times \QProc {\times} \QProc \red \Proc.\] 
The guiding intuition is that we construct alongside the process a distributed memory
allocator, the process' access to which is mediated through the second argument
to the function, called $p$ below. The first argument, called $n$ below, determines the shape of the memory for the given allocator.

Since Yoshida's calculus is parametric in a set of names, we can
choose $\QProc$ for that set.  Given a process $P$ in Yoshida's
calculus, we pick names $n$ and $p$ in the RHO calculus such that $n \neq p$
and both are distinct from the free names of $P$.  Then we define

\begin{equation*}
  \meaningof{P} = \meaningof{P}_2(n, p),
\end{equation*}

Name allocation will make heavy use of the following two name
constructors

\begin{eqnarray*}
  x^l & := & \quotep{\mathsf{b}_{\mathsf{l}}(x,\mathsf{m}(x,*x))} \\
  x^r & := & \quotep{\mathsf{b}_{\mathsf{r}}(x,\mathsf{m}(x,*x))}
\end{eqnarray*}

Note that by construction, $\quotep{P}$ cannot occur as a name in $P$
and hence any name derived from a process that is built using
$\quotep{P}$ cannot occur in $P$. Thus, the effect of the superscripts
$l$ and $r$ on a name $x$ is to construct a name that is guaranteed to
be fresh with respect to the free names of the process being
interpreted. More generally, mentioning a name, say $n$, in the
constructor of a another name, say $n'$, guarantees distinction
between $n$ and $n'$; likewise, mentioning a process, say $P$, in the
constructor of a name $n$ guarantees that $n$ is fresh in $P$. The
particular choices of combinators used in the name
constructors are irrelevant for the purposes of freshness. We
make heavy use of this fact in our interpretation of prefix
elimination.

The interpretation function $\meaningof{-}_2(n, p)$ is straightforward
for all but replication and $\mathsf{new}$.

\begin{eqnarray*}
    \meaningof{\pzero}_2 (n,p)
      & = &
       \pzero \\
    \meaningof{\emph{c}(\vect{a})}_2 (n,p) 
      & = & 
      \emph{c}(\vect{a})\mbox{, where $\emph{c}$ is any combinator but $m$} \\
    \meaningof{\mathsf{m}(a,\quotep{P})}_2 (n,p) 
      & = & 
          \mathsf{m}(a, \quotep{\meaningof{P}}_2 (n,p)) \\
    \meaningof{P \juxtap Q}_2 (n,p) 
      & = & 
    \meaningof{P}_2 (n^l, p^l)
         \juxtap \meaningof{Q}_2 (n^r, p^r) \\ 
\end{eqnarray*}

These latter two forms require extra care. We define them in terms of
a prefix form and then use a version of Yoshida's prefix elimination
to remove the prefix.

\begin{eqnarray*}
    \meaningof{\mathsf{*} P}_2 (n,p)
          & = & \binpar{\mathsf{m}(x, \quotep{\meaningof{P}_3(n^r,p^r)})}
                  {\binpar{D(x,v,w)}
                    {\binpar{\mathsf{m}(n^r, *n^l)}{\mathsf{m}(n^r, *n^l)}}} \\
                  & & \mbox{where } 
                      x = @(\meaningof{P}_2(n,p))^{ll}, 
                      v = @(\meaningof{P}_2(n,p))^{lr}, \mbox{ and }
                      w = @(\meaningof{P}_2(n,p))^{rr} \\
    \meaningof{(\new \; x ) P}_2 (n, p) 
          & = & 
         \meaningof{\prefix{p}{x}{\binpar{\meaningof{P}_2 ( n^l, p^l )}{\mathsf{m}(p, n)}}}_4(n, p)
\end{eqnarray*}

As expected, the interpretation of replication makes use of the $D$
operator defined above. Note that name allocation is intertwined with
prefix and new elimination.
         
\begin{eqnarray*}
  \meaningof{P}_3(n, p) 
    & := & 
      \meaningof{\prefix{n}{n'}{\prefix{p}{p'}{(\binpar{\meaningof{P}_2(n',p')}
        {(\binpar{D(x)}{\binpar{\outputp{n}{n'^l}}{\outputp{p}{p'^l}}})})}}}_4. \\
\end{eqnarray*}

To handle prefix elimination we must import most of Yoshida's
algorithm. The key difference is that we must allocate names that
guarantee freshness relative to the free names of the processes being
translated. In those rules below with a ``where'' clause, the specific
choice of combinators in the names is not so important as mentioning
those names and processes with respect to which the name mush be fresh.
For example, in rule $I$, for prefix to a parallel
composition, we must ensure that $v$ and $w$ are fresh with respect to
the names in $P$ and $Q$, and distinct from each other, and then we
must update both the name allocator for each parallel component and
the channels on which fresh names are communicate (so that there is no
interference between the two components) in the recursive call. 

Because the input to $\meaningof{-}_4$ is already in combinator form,
we do not have to import rules 2 -- 4 of her algorithm. Like her, we
assume the following annotations ($+$ stands for the output and $-$
stands for the input), which denote how each name is used in the rules
of interaction:

\[\mathsf{m}(a^{+},v^{\pm});\mathsf{d}(a^{-},b^{+},c^{+});\mathsf{k}(a^{-});\mathsf{fw}(a^{-},b^{+});\mathsf{b}_{\mathsf{l}}(a^{-}b^{+});\mathsf{b}_{\mathsf{r}}(a^{-}b^{-});\mathsf{s}(a^{-}b^{-}c^{+})\]

As above the annotated polarities are preserved by reduction, e.g.
\[\binpar{\mathsf{d}(a^{-},b^{+},c^{+})}{\mathsf{m}(a^{+},v)} \to \binpar{\mathsf{m}(b^{+},v)}{\mathsf{m}(b^{+},v)}\]

Similarly, for economy of expression, we emulate Yoshida's use of
$\emph{c}$ to represent any combinator matching the arity
specification. 

\[\begin{array}{llrl}
(I)&  \meaningof{\prefix{p}{x}{\binpar{P}{Q}}}_4(n, q) 
    & := & 
    (\mathsf{d}(p,v,w)|\binpar{\meaningof{\prefix{v}{x}{P}}_{4}(n_{1}, q_{1})}{\meaningof{\prefix{w}{x}{Q}}_{4}(n_{2}, q_{2}) } \\
    & & & \mbox{where $v = \quotep{(\mathsf{m}(q,\quotep{\binpar{\mathsf{b}_{\mathsf{l}}(q,n)}{\binpar{P}{Q}}}))}$, $w = \quotep{(\mathsf{m}(q,\quotep{\binpar{\mathsf{b}_{\mathsf{r}}(q,n)}{\binpar{P}{Q}}}))}$,} \\
    & & &\mbox{$n_{1} = \quotep{(\binpar{\mathsf{b}_{\mathsf{l}}(v,w)}{\mathsf{m}(q,\quotep{\mathsf{m}(v,*w)})})}$, $n_{2} = \quotep{(\binpar{\mathsf{b}_{\mathsf{r}}(v,w)}{\mathsf{m}(q,\quotep{\mathsf{m}(v,*w)})})}$} \\
    & & &\mbox{$q_{1} = \quotep{(\binpar{\mathsf{b}_{\mathsf{l}}(n_{1},n_{2})}{\mathsf{m}(q,\quotep{\mathsf{m}(v,*w)})})}$, $q_{2} = \quotep{(\binpar{\mathsf{b}_{\mathsf{r}}(n_{1},n_{2})}{\mathsf{m}(q,\quotep{\mathsf{m}(v,*w)})})}$} \\
(V)&  \meaningof{\prefix{p}{x}{\emph{c}(v^{+},w)}}_4(n, q) 
    & := & 
    \mathsf{s}(p,a,v)|\emph{c}(a^{+},\vect{w})
    \mbox{, where $a = @(\binpar{\mathsf{m}(q,n)}{\emph{c}(v^{+},\vect{w})})$ and $x \not\in v : \vect{w}$}\\
(VI)&  \meaningof{\prefix{p}{x}{\emph{c}(v^{-},\vect{w})}}_4(n, q) 
    & := & 
    \mathsf{s}(p,v,a)|\emph{c}(a^{-},\vect{w})
    \mbox{, where $a = @(\mathsf{m}(q,n)|\emph{c}(v^{-},w))$ and $x \not\in v : \vect{w}$} \\
(VII)&  \meaningof{\prefix{p}{x}{\mathsf{m}(v,x)}}_4(n, q) 
    & := & 
    \mathsf{fw}(p,v) \\
(VIII)&  \meaningof{\prefix{p}{x}{\mathsf{fw}(x,v)}}_4(n, q) 
    & := & 
    \mathsf{b}_{\mathsf{l}}(p,v) \\
(IX)&  \meaningof{\prefix{p}{x}{\mathsf{fw}(v,x)}}_4(n, q) 
    & := & 
    \mathsf{b}_{\mathsf{r}}(p,v) \\
(X)&  \meaningof{\prefix{p}{x}{\emph{c}(\vect{v},x^{+},\vect{w})}}_4(n, q) 
    & := & 
    \meaningof{\prefix{p}{x}{\binpar{\mathsf{fw}(a,x)}{\emph{c}(\vect{v},a^{+},\vect{w})}}}_4(n', q)\\
    & & & \mbox{where $a = @(\binpar{\mathsf{m}(q,n)}{\emph{c}(\vect{v},x^{+},\vect{w})})$} \mbox{ and $n'= \quotep{(\mathsf{m}(a,\quotep{\mathsf{m}(q,*n)}))}$} \\
(XI)&  \meaningof{\prefix{p}{x}{\emph{c}(x^{-},\vect{v})}}_4(n, q) 
    & := & 
    \meaningof{\prefix{p}{x}{\binpar{\mathsf{fw}(x,a)}{\emph{c}(a^{-},\vect{v})}}}_4(n', q)\\
    & & & \mbox{where $a = @(\binpar{\mathsf{m}(q,n)}{\emph{c}(x^{-},v)})$} \mbox{ and $n'= \quotep{(\mathsf{m}(a,\quotep{\mathsf{m}(q,*n)}))}$} \\
(XII)&  \meaningof{\prefix{p}{x}{\mathsf{b}_{\mathsf{r}}(v,x^{-})}}_4(n, q) 
    & := & 
    \meaningof{\prefix{p}{x}{(\binpar{\mathsf{d}(v,w_{1},w_{2})}{\binpar{\mathsf{s}(w_{1},x,w_{3})}{\mathsf{b}_{\mathsf{r}}(w_{2},w_{3})}})}}_{4}(n', q) \\
    & & & \mbox{where $w_{1} = \quotep{(\binpar{\mathsf{b}_{\mathsf{l}}(q,n)}{\mathsf{m}(q,\quotep{\mathsf{b}_{\mathsf{r}}(v,x^{-})})})}$, $w_{2} = \quotep{(\binpar{\mathsf{b}_{\mathsf{r}}(q,n)}{\mathsf{m}(q,\quotep{\mathsf{b}_{\mathsf{r}}(v,x^{-})})})}$,} \\
    & & & \mbox{$w_{3} = \quotep{(\binpar{\mathsf{b}_{\mathsf{l}}(p,v)}{\mathsf{m}(w_{1},w_{2}})}$, and $n' = \quotep{(\mathsf{s}(w_{1},w_{2},w_{3}))}$} \\
(XIII)&  \meaningof{\prefix{p}{x}{\mathsf{s}(v,x^{-},w)}}_4(n, q) 
    & := & 
    \meaningof{\prefix{p}{x}{(\binpar{\mathsf{s}(v,w_{1},w_{2})}{\binpar{\mathsf{m}(w_{1},x)}{\mathsf{b}_{\mathsf{l}}(w_{2},w)}})}}_{4}(n', q) \\
    & & & \mbox{where $w_{1} = @(\binpar{\mathsf{b}_{\mathsf{l}}(q,n)}{\mathsf{s}(v,x^{-},w)})$, $w_{2} = @(\binpar{\mathsf{b}_{\mathsf{r}}(q,n)}{\mathsf{s}(v,x^{-},w)})$,} \\
    & & & \mbox{and $n' = \quotep{(\mathsf{m}(w_{1},w_{2}))}$}
\end{array}\]

Note that all $\new$-binding is now interpreted, as in Wischik's
global $\pi$-calculus, as an input \cite{globalpi}.

It is also noteworthy that the translation is dependent on how the
parallel compositions in a process are associated. Different
associations will result in different bindings for $\new$
names. This will not result in different behavior, however:
while the RHO combinators can encode different behaviors depending
on the choice of name, Yoshida's combinators cannot and the
embedding is insensitive to the choice.

\subsubsection{Faithfulness of the translation}

\begin{definition}
An \emph{observation relation}, $\downarrow_{\mathcal{N}}$, over a set
of names, $\mathcal N$, is the smallest relation satisfying the rules
below.

\infrule[Out-barb]{y \in {\mathcal N}, \; x \nameeq y}
    {\mathsf{m}(x-) \downarrow_{\mathcal{N}} x}

\infrule[Par-barb]{\mbox{$P\downarrow_{\mathcal{N}} x$ or $Q\downarrow_{\mathcal N} x$}}
    {\binpar{P}{Q} \downarrow_{\mathcal{N}} x}

We write $P \Downarrow_{\mathcal{N}} x$ if there is $Q$ such that 
$P \wred Q$ and $Q \downarrow_{\mathcal{N}} x$.
\end{definition}

This definition is parametric in the the argument accepted in the
second position in the message combinator, $\mathsf{m}(x-)$, {\em i.e.} the payload
of the message: in the RHO combinators the payload is a process,
while in Yoshida's the payload is a name. Likewise, because the
definition of barbed bisimulation given below is dependent on the
definition of the observation relation, the definition is really a
template for the notion of bisimulation that must be instantiated to
the kind of payload accepted by the message combinator.


\begin{definition}
An  ${\mathcal N}$-\emph{barbed bisimulation} over a set of names, ${\mathcal N}$, is a symmetric binary relation 
${\mathcal S}_{\mathcal N}$ between agents such that $P{\mathcal S}_{\mathcal N}Q$ implies:
\begin{enumerate}
\item If $P \red P'$ then $Q \wred Q'$ and $P'{\mathcal S}_{\mathcal N} Q'$.
\item If $P\downarrow_{\mathcal N} x$, then $Q\Downarrow_{\mathcal N} x$.
\end{enumerate}
$P$ is ${\mathcal N}$-barbed bisimilar to $Q$, written
$P \wbbisim_{\mathcal N} Q$, if $P {\mathcal S}_{\mathcal N} Q$ for some ${\mathcal N}$-barbed bisimulation ${\mathcal S}_{\mathcal N}$.
\end{definition}

\begin{theorem}
  $P \wbbisim_{\pi} Q \iff \ldb P \rdb \wbbisim_{\texttt{FN}(P)} \ldb Q \rdb$.
\end{theorem}

\begin{proof}[Proof sketch]
  The forward direction $\Rightarrow$ is immediate from the definition
  of the translation. The reverse direction is only interesting in the
  case of $!$ and $\mathsf{new}$. The replication case follows immediately
  from the calculation following definition \ref{replication}. 
  In the $\new$ case, transitions on $\mathsf{new}$-bound
  names will be in one-to-one correspondence with names provided by
  the name parameters of the translation function. By construction,
  these are not observable by the observation relation.
\end{proof}

\begin{remark}
  In light of this theorem, it is worth pointing out that this version
  of the RHO combinators has no rule for \emph{introducing} terms of
  the form $\dropn{x}$. The $b_r$ and $b_l$ combinators introduce new
  names from processes, but the do not introduce new reflection
  terms. Yet, this calculus suffices to faithfully represent the
  Yoshida combinators. This is because the translation function is
  carefully introducing just those terms, via the $D(x,v,w)$ operator,
  guided by the use of replication in the source to the
  translation. 
\end{remark}

\section{Conclusion and future work}
We have shown how to construct a concurrent higher-order combinator
calculus that uses reflection to avoid the necessity for new and bound
names.  Yoshida's combinators, and therefore the asynchronous {\pic},
have a faithful embedding into the calculus. While the results are
interesting in their own right we developed them to serve a larger
purpose. In a forthcoming paper we demonstrate an algorithm taking a
graph-enriched Lawvere theory (representing a formal specification of
a term calculus) and a more vanilla Lawvere theory (representing a
specification of a notion of collection, such as set, or bag, or list,
etc) and produce a type system for the term calculus enjoying
soundness and completeness. Nominal phenomena do not neatly fit inside
the expressive power of graph-enriched Lawvere theories, thus
potentially limiting the scope of the applicability of this
algorithm. However, with the reflective techniques we can extend this
algorithm to cover many languages and term calculi with binding.

At a more foundational level it turns out that the reflective
techniques can be applied to set theory itself deriving a form of FM
set theory in which the ``atoms'' of one set theory are the sets of
another copy of set theory. Roughly speaking, we may imagine that
there are ``red'' sets and ``black'' sets. The ``term constructors''
of red set theory (the curly braces of set notation) can use either
red sets or black sets, which the element-of predicate ($\in$) red set
theory can only inspect red sets and treats black sets as
``atomic''. The symmetric situation holds for black set theory: the
element-of predicate cannot see inside red sets. This framework all
the building blocks necessary to realize a version of FM set theory
from a more traditional set theory with reflection and provides a
bridge between Gabbay and Pitts' account of nominal phenomena and the
reflection-based approach.

\bibliographystyle{amsplain}
\bibliography{rhocomb}

\section{Appendix: Translating the {\pic} into the {\rhoc}}
\subsection{\rhoc}

It is striking to compare the {\rhoc} with the {\pic} as the former is
a vastly simpler theory and yet has enjoys more
features. Specifically, the specification of the {\rhoc}'s structural
equivalence and reduction rules are notably simpler, with one small
technical caveat: name equivalence depends on structural equivalence
which depends upon $\alpha$-equivalence which depends on name
equivalence. Meredith and Radestock show that this cycle terminates
innocently due to the design of the grammar. This technical complexity
seems a small price to pay for a much simpler calculus that enjoys
both higher order communication as well reflection and meta-programming
features.

\begin{mathpar}
\inferrule* [lab=process] {} {P, Q \bc \pzero \;\bm\; \mathsf{for}( y
  \leftarrow x )P \;\bm\; x!(Q) \;\bm\;	\mathsf{*}x \;\bm\; P|Q }
\and
\inferrule* [lab=name] {} {x, y \bc \mathsf{@}P }
\end{mathpar}

\begin{definition}
\emph{Free and bound names} The calculation of the free names of a
process, $P$, denoted $\freenames{P}$ is given recursively by

\begin{mathpar}
  \freenames{\pzero} = \emptyset
  \and
  \freenames{\mathsf{for}(y \leftarrow x)(P)} = \{ x \} \cup \freenames{P}\setminus\{y\}
  \and
  \freenames{x!(P)} = \{ x \} \cup \freenames{P}
  \and
  \freenames{P|Q} = \freenames{P} \cup \freenames{Q}
  \and
  \freenames{\mathsf{}{x}} = \{ x \} \\
\end{mathpar}

An occurrence of $x$ in a process $P$ is \textit{bound} if it is not
free. The set of names occurring in a process (bound or free) is
denoted by $\names{P}$.
\end{definition}

\begin{definition}
  The {\em structural congruence} $\equiv$
  between processes \cite{SangiorgiWalker} is the least congruence containing
  alpha-equivalence and satisfying the commutative monoid laws
  (associativity, commutativity and $\pzero$ as identity) for parallel
  composition $|$.
\end{definition}

\begin{definition}
  The {\em name equivalence} $\nameeq$ is the least congruence
  satisfying these equations
  \begin{mathpar}
  \inferrule*[lab=Quote-drop] {}{ \quotep{\dropn{x}} \nameeq x }
  \and
  \inferrule*[lab=Struct-equiv] { P \scong Q } { \quotep{P} \nameeq \quotep{Q} }
  \end{mathpar}
\end{definition}

\subsection{Operational semantics} 

\begin{mathpar}
  \inferrule* [lab=COMM] {x_{trgt} \nameeq x_{src}} {\mathsf{for}( y \leftarrow x_{trgt} )P | x_{src}!(Q)
  \red P\substn{\mathsf{@}Q}{y}}
  \and
  \inferrule* [lab=PAR]{P \red P'}{P|Q \red P'|Q}
  \and
  \inferrule* [lab=EQUIV]{{P \scong P'} \andalso {P' \red Q'} \andalso {Q' \scong Q}}{P \red Q}
\end{mathpar}

\subsection{\pic}

In this presentation of the {\pic} we update the syntax for
input-guarded processes to reflect the widespread adoption of
comprehension notation in languages ranging from Scala to Python for
use in reactive programming. Here we go with the Scala notation
writing $\mathsf{for}( y \leftarrow x )P$ where Milner might have
written $x?(y)P$. Admittedly, it's somewhat more verbose, but conveys
to a younger generation of programmers more familiar with reactive
programming the intended semantics of the expression. Similarly, since
we keep the output expression $x!(y)$ we avoid collision with the
traditional notation for reflection ($!P$) by writing $\mathsf{*}P$
instead, which is at least somewhat reminiscent of the Kleene star and
the familiar from regular expressions.

\begin{mathpar}
\inferrule* [lab=process] {} {P, Q \bc \pzero \;\bm\; \mathsf{for}( y
  \leftarrow x )P \;\bm\; x!(y) \;\bm\; (\mathsf{new}\;x)P \;\bm\; P|Q \;\bm\;	\mathsf{*}P}
\end{mathpar}

Note that there is no production for $x$'s and $y$'s in the
grammar. This reflects the fact that the {\pic} is \emph{parametric}
in the collection channel names. That collection merely has to be
countably infinite and have an effective equality. As such it is
perfectly reasonable to choose a collection names, namely the names of
the {\rhoc}. We can make this choice without loss of generality
because we can always choose some other countably infinite set with an
effective equality, say $\mathcal{N}$ and then require an invertible
map, $\mathsf{code} : \mathcal{N} \to @\mathsf{Proc}$.

\subsection{Structural congruence}

\begin{definition}
The {\em structural congruence}, $\equiv$, between processes is 
the least congruence closed with respect to
alpha-renaming, satisfying the abelian monoid laws for 
parallel (associativity, commutativity and $\pzero$ 
as identity), and the following axioms:
\begin{enumerate}
\item the scope laws:
\begin{eqnarray*}
 (\mathsf{new}\;x)\pzero & \equiv & \pzero, \\
 (\mathsf{new}\;x)(\mathsf{new}\;x)P & \equiv & (\mathsf{new}\;x)P, \\
 (\mathsf{new}\;x)(\mathsf{new}\;y)P & \equiv & (\mathsf{new}\;y)(\mathsf{new}\;x)P, \\
 P|(\mathsf{new}\;x)Q & \equiv & (\mathsf{new}\;x)P|Q, \; \mbox{\textit{if} }x \not\in \freenames{P} 
\end{eqnarray*}
\item
the recursion law:
\begin{eqnarray*}
 \mathsf{*}P \equiv P|\mathsf{*}P
\end{eqnarray*}
\item
the name equivalence law:
\begin{eqnarray*}
 P \equiv P \substn{x}{y}, \; \mbox{\textit{if} }x \nameeq y
\end{eqnarray*}
\end{enumerate}
\end{definition}

\subsection{Operational semantics} 

\begin{mathpar}
  \inferrule* [lab=COMM] {} {\mathsf{for}( y \leftarrow x )P | x!(v) \red P\substn{v}{y}}
  \and
  \inferrule* [lab=PAR]{P \red P'}{P|Q \red P'|Q}
  \and
  \inferrule* [lab=NEW]{P \red P'}{(\mathsf{new}\;x)P \red (\mathsf{new}\;x)P'}
  \and
  \inferrule* [lab=EQUIV]{{P \scong P'} \andalso {P' \red Q'} \andalso {Q' \scong Q}}{P \red Q}
\end{mathpar}

Here we stick with tradition and write $\wred$ for $\red^*$, and rely
on context to distinguish when $\red$ means reduction in the {\pic}
and when it means reduction in the {\rhoc}. The set of {\pic}
processes will be denoted by $\Proc_{\pi}$.

\subsection{The translation}

The translation will be given by a function, $\meaningof{-}( -, - ) :
\Proc_{\pi} \times \QProc \times \QProc \red \Proc$. The guiding
intuition is that we construct alongside the process a distributed memory
allocator, the process' access to which is mediated through the second argument
to the function. The first argument determines the shape of the memory
for the given allocator.

Given a process, $P$, we pick $n$ and $p$ such that $n \neq p$ and
distinct from the free names of $P$. For example, $n = \quotep{\Pi_{m
\in \freenames{P}}\outputp{m}{\quotep{\pzero}}}$ and $p =
\quotep{\Pi_{m \in
\freenames{P}}\mathsf{for}({\quotep{\pzero}} \leftarrow {m}){\pzero}}$. Then

\begin{equation*}
	\meaningof{P} = \meaningof{P}_{2nd}( n, p )
\end{equation*}

where

\begin{eqnarray*}
   	\meaningof{\pzero}_{2nd} (n,p) & = & \pzero \\
   	\meaningof{x!(@Q)}_{2nd} (n,p) & = & x!(Q) \\
   	\meaningof{\mathsf{for}( y \leftarrow x) P}_{2nd} (n,p) 
   		& = & 
 		\mathsf{for}( y \leftarrow x ) \meaningof{P}_{2nd} (n,p) \\
   	\meaningof{P | Q}_{2nd} (n,p) 
   		& = & 
 		\meaningof{P}_{2nd} (n^{l},p^{l}) | \meaningof{Q}_{2nd} (n^{r},p^{r}) \\
   	\meaningof{\mathsf{*} P}_{2nd} (n,p)
   		& = & x!(\meaningof{P}_{3rd}( n^{r}, p^{r} ))|D(x)|n^{r}!n^{l}|p^{r}!p^{l} \\
   	\meaningof{(\mathsf{new} \; x) P}_{2nd} (n,p) 
   		& = & 
 		\mathsf{for}(x \leftarrow p)\meaningof{P}_{2nd} ( n^{l}, p^{l} )|p!(n) \\
\end{eqnarray*}

and

\begin{eqnarray*}
	x^{l} & \triangleq & \quotep{\outputp{x}{x}} \\
	x^{r} & \triangleq & \quotep{\prefix{x}{x}{\pzero}} \\
	\meaningof{P}_{3rd}( n'', p'' ) 
		& \triangleq & 
			\prefix{n''}{n}{\prefix{p''}{p}{(\binpar{\meaningof{P}_{2nd}(  n, p )}
							        {(\binpar{D(x)}{\binpar{\outputp{n''}{n^{l}}}{\outputp{p''}{p^{l}}}})})}} \\
\end{eqnarray*}

\begin{remark}
	Note that all $\mathsf{new}$-binding is now interpreted, as in Wischik's
	global $\pi$-calculus, as an input guard \cite{globalpi}.
\end{remark}
	
\begin{remark}
	It is also noteworthy that the translation is dependent on how
	the parallel compositions in a process are
	associated. Different associations will result in different
	bindings for $\mathsf{new}$-ed names. This will not result in different
	behavior, however, as the bindings will be consistent
	throughout the translation of the process.
\end{remark}

\begin{theorem}[Correctness]	
	$P \wbbisim_{\pi} Q \iff \ldb P \rdb \wbbisim_{r(\texttt{FN}(P))} \ldb Q \rdb$.
\end{theorem}

\emph{Proof sketch}: An easy structural induction.

One key point in the proof is that there are contexts in the {\rhoc}
that will distinguish the translations. But, these are contexts that
can see the fresh names, $n$, and the communication channel, $p$, for
the `memory allocator'. These contexts do not correspond to any
observation that can be made in the {\pic} and so we exclude them in
the {\rhoc} side of our translation by our choice of ${\mathsf N}$
for the bisimulation. This is one of the technical motivations behind
our introduction of a less standard bisimulation.

\begin{example}
	In a similar vein consider, for an appropriately chosen $p$ and $n$ we have
	\begin{equation*}
		\meaningof{(\mathsf{new}\;v)(\mathsf{new}\;v) u!(v)} = \mathsf{for}(v \leftarrow p)(\mathsf{for}({v} \leftarrow {\quotep{p!(p)}})(u!(v)|\quotep{p!(p)})!(\quotep{n!(n)})) | p!(n)
	\end{equation*}
	and
	\begin{equation*}
		\meaningof{(\mathsf{new}\;v)u!(v)} = \mathsf{for}(v \leftarrow p)(u!(v) )|p!(n)
	\end{equation*}

	Both programs will ultimately result in an output of a single
	fresh name on the channel $u$. But, the former program will
	consume more resources. Two names will be allocated; two memory
	requests will be fulfilled. The {\rhoc} can see this, while the
	{\pic} cannot. In particular, the {\pic} requires that
	$(\mathsf{new}\;x)(\mathsf{new}\;x)P \equiv (\mathsf{new} x)P$.

	Implementations of the {\pic}, however, having the property that
	$(\mathsf{new}\;x)P$ involves the allocation of memory for the
	structure representing the channel $x$ come to grips with the
	implications this requirement has regarding memory management. If
	memory is allocated upon encountering the $\mathsf{new}$-scope, there are
	situations where the left-hand side of the equation above will
	fail while the right-hand will succeed. Remaining faithful to the
	equation above requires that such implementations are
	\textit{lazy} in their interpretation of $(\mathsf{new}\;x)P$, only
	allocating the memory for the fresh channel at the first moment
	when that channel is used.

	Having a detailed account of the structure of names elucidates
	this issue at the theoretical level and may make way to offer
	guidance to implementations.
\end{example}

\end{document}